\documentclass[fleqn]{llncs}

\usepackage{hyperref}

\usepackage{times}
\usepackage{amssymb}
\usepackage{amsmath}
\usepackage{url}

\usepackage{xspace}
\usepackage[all]{xy}

\renewcommand{\tt}{\ttfamily}
\newcommand{\codefont}{\small\tt}
\newcommand{\code}[1]{\mbox{\codefont{#1}}}
\newcommand{\ccode}[1]{``\code{#1}''}
\newcommand{\fcode}[1]{\mbox{\codefont{\footnotesize{#1}}}} 
\newcommand{\us}{\raise-.8ex\hbox{-}}
\newcommand{\xtilde}{\!\raise-.75ex\hbox{\char`\~}} 
\newcommand{\true}{\code{True}\xspace}

\newcommand{\ctxt}[1]{{\cal{C}}[#1]} 
\def\deriv#1{\buildrel #1 \over \to}
\newcommand{\tostar}{\mathrel{\deriv{*}}}
\newcommand{\rulename}[1]{\mbox{\sf #1}}

\usepackage{listings}
\lstnewenvironment{curry}{\lstset{literate={->}{{$\rightarrow{}\!\!\!$}}3}}{}
\lstnewenvironment{prolog}{}{}
\newcommand{\listline}{\vrule width0pt depth1.5ex}

\begin{document}

\title{Equivalence Checking of Non-deterministic Operations\thanks{%
This is an extended version of a paper appeared in
Proc.\ of the 14th International Symposium on Functional and Logic Programming
(FLOPS 2018), Springer LNCS 10818, pp. 149-165, 2018}}

\author{
Sergio Antoy\inst{1}
\kern1em
Michael Hanus\inst{2}
}
\institute{
Computer Science Dept., Portland State University, Oregon, U.S.A.\\
\email{antoy@cs.pdx.edu}\\[1ex]
\and
Institut f\"ur Informatik, CAU Kiel, D-24098 Kiel, Germany. \\
\email{mh@informatik.uni-kiel.de}
}

\maketitle

\begin{abstract}
Checking the semantic equivalence of operations is an important
task in software development.
For instance, regression testing is a routine task performed when
software systems are developed and improved,
and software package managers require the equivalence
of operations in different versions of a package within the same
major number version.
A solid foundation is required
to support a good automation of this process.
It has been shown that the notion of equivalence
is not obvious when non-deterministic features are present.
In this paper, we discuss a general notion of equivalence
in functional logic programs
and develop a practical method to check it.
Our method is integrated in a property-based testing tool
which is used in a software package manager to check
the semantic versioning of software packages.
\end{abstract}

\section{Motivation}
\label{sec:intro}

Functional logic languages combine the most important
features of functional and logic programming in a single language
(see \cite{AntoyHanus10CACM,Hanus13} for recent surveys).
In this paper we consider Curry \cite{Hanus16Curry},
a contemporary functional logic language that
conceptually extends Haskell with common features of logic programming.
Hence, a Curry function is evaluated lazily and may
be non-deterministic when defined by overlapping rules.
As discussed in \cite{BacciEtAl12},
the combination of these features raises new issues
for defining the equivalence of expressions.
Actually, three different notions of equivalence can be distinguished:
\begin{enumerate}
\item \emph{Ground equivalence}: Two expressions are equivalent
if they produce the same results when their variables are replaced
by ground terms.
\item \emph{Computed-result equivalence}:
Two expressions are equivalent if they produce the same outcomes,
i.e., variables in expressions are considered as free variables
which might be instantiated during the evaluation process.
\item \emph{Contextual equivalence}:
Two expressions are equivalent if they produce the same outcomes
in all possible contexts.
\end{enumerate}
Ground equivalence seems reasonable for functional programs
since free variables are not allowed in
expressions to be evaluated in functional programming.
For instance, consider the Boolean negation defined by:
\begin{curry}
not False = True
not True  = False
\end{curry}
The expressions \code{not$\,$(not$\;$x)} and \code{x} are ground
equivalent, which can be checked easily by instantiating \code{x}
to \code{True} and \code{False}, respectively, and evaluating
both expressions.
However, these expressions are not computed-result equivalent
w.r.t.\ the narrowing semantics of functional logic programming.
The expression \code{not$\,$(not$\;$x)} evaluates to the following
two outcomes.
The first outcome is the value \code{False} where the variable
\code{x} is bound to \code{False} which we compactly represent as
\code{\{x=False\}$\;$False}.
The second outcome is \code{\{x=True\}$\;$True}.
By contrast, the expression \code{x} evaluates to the single result
\code{\{\}$\;$x} without instantiating the free variable \code{x}.
Due to these differences, Bacci et al.\ \cite{BacciEtAl12} state
that ground equivalence is ``the (only possible) equivalence notion
used in the pure functional paradigm.''
As we will see later, this is not true since contextual equivalence
is also relevant in non-strict functional languages.

The previous example shows that the evaluation
of ground equivalent expressions might result in answers
with different degrees of instantiation.
However, the presence of logic variables and non-determinism
might also lead to different results when ground equivalent
expressions are put in some context as the following contrived example
shows \cite{BacciEtAl12} (a more natural example will be shown later).
\begin{example}
  \label{contrived-ground-equiv}
Consider the following functions:
\begin{curry}
f x = C (h x)                     g A = C A
h A = A
\end{curry}
The expressions \code{f$\;$x} and \code{g$\;$x} are
computed-result equivalent since the only computed result
is \code{\{x=A\}$\;$C$\;$A}.
Now consider the following operation:
\begin{curry}
k (C x) B = B
\end{curry}
Below we show the evaluations of \code{k$\,$(f$\;$x)$\;$x}
and \code{k$\,$(g$\;$x)$\;$x} in which the subscript of ``$\to$''
shows the binding of the step:
\begin{curry}
k (f x) x $~\to_\mathtt{\{\}}~~~\,$ k (C (h x)) x $~\to_\mathtt{\{x=B\}}~$ B
k (g x) x $~\to_\mathtt{\{x=A\}}~$ k (C A) A 
\end{curry}
Although \code{f$\;$x} and \code{g$\;$x} produce exactly the
same set of values, namely the singleton \code{\{A\}},
they produce different outcomes within some context.
This shows that there are ground and computed-result equivalent expressions
that are not contextually equivalent.
The reason of this behavior is the evaluation strategy which is
non-strict (lazy) and binds free variables (narrows).
More specifically, \code{h$\;$x} is not evaluated (it is ignored)
in the first computation
and \code{x} is bound to \code{A} when \code{g$\;$x} is evaluated
in the second computation.
\end{example}
The equivalence of operations is important
when existing software packages are maintained, e.g.,
refactored or implemented with more efficient data structures.
In this case, we want to ensure that operations available
in the API of both versions of a software package are equivalent,
unless the intent was to change the API.
For this purpose, software package management systems
associate version numbers to software packages.
In the semantic versioning standard,\footnote{\url{http://www.semver.org}}
a version is identified by major, minor, and patch numbers,
separated by dots, and an optional pre-release specifier.
For instance, \code{2.0.1} and \code{3.2.1-alpha.2} are
valid version identifiers.
An intended and incompatible change of API operations
is marked by a change in the major version number.
Thus, operations available in two versions of a package with
identical major version numbers should be equivalent.
Unfortunately, most package managers do not check
this equivalence but leave it as a recommendation to the
package developer.

Improving this situation is the motivation for our work.
We want to develop a tool for checking the equivalence
of two operations.
Since we aim at integrating this kind of semantic versioning
checking in a practical software package manager
\cite{Hanus17ICLP}, the tool should be fully automatic.
Thus, we are going to test equivalence properties rather
than to verify them. Although this might be unsatisfactory
from a theoretical point of view, it could be quite powerful
from a practical point of view and might prevent wasting time
that would go into attempting to prove properties that the tool
has already disproved.
For instance, property-based test tools like
QuickCheck \cite{ClaessenHughes00} provide great
confidence in programs by checking program properties with many
test inputs.
For instance, we could check the equivalence of
two operations $f$ and $f'$ by checking the equation
$f\;x = f'\;x$ with many values for $x$.
The previous discussion of equivalence criteria
shows that this property checks only the ground equivalence
of $f$ and $f'$.
However, in the context of semantic versioning checking,
ground equivalence is too restricted since
equivalent operations should deliver the same results
in any context.
Therefore, contextual equivalence is desired.
Actually, this kind of equivalence has been proposed
in \cite{AntoyHanus12PADL} as the only notion
to state the correctness of an implementation w.r.t.\ a
specification in functional logic programming.
Unfortunately, the automatic checking of contextual equivalence
with property-based test tools does not seem feasible
due to the boundless number of possible contexts.
Therefore, Bacci et al.\ \cite{BacciEtAl12} state:
``In a test-based approach\ldots{}the addition of a further outer context
would dramatically alter the performance.''
Consequently, the authors abandon the use of a standard
property-based test tool in their work.

In this paper we show that we can use such tools for
contextual equivalence (and, thus, semantic versioning) checking
if we use an appropriate encoding of test data.
For this purpose, we develop some theoretical results
that allow us to reduce the contexts to be considered
for equivalence checking.
From these results, we show how property-based testing
can be used for this purpose.
Based on these results, we extend an existing
property-based test tool for functional logic programs
\cite{Hanus16LOPSTR} to test the equivalence of operations.
This is the basis of a software package manager
with semantic versioning checking \cite{Hanus17ICLP}.

In the next section, we review the main concepts of
functional logic programming and Curry.
Section~\ref{sec:equiv} defines our notion of equivalence
which is used in Sect.~\ref{sec:refining}
to develop practically useful characterizations of equivalent operations.
Section~\ref{sec:currycheck} shows how to use these criteria
in a property-based testing tool.
This tool is evaluated in Sect.~\ref{sec:evaluation}
on various examples.
Section~\ref{sec:applications} describes two applications
of our notion of equivalence: consistency of package versions
and consistency of implementations w.r.t.{} their specifications.
Section~\ref{sec:verification} shows the verification
of an equivalence property where our results are helpful.
Section~\ref{sec:related} discusses some related work
before we conclude.

\section{Functional Logic Programming and Curry}
\label{sec:flp}

We briefly review those elements of functional logic languages
and Curry that are necessary to understand the contents of this paper.
More details can be found in surveys on
functional logic programming \cite{AntoyHanus10CACM,Hanus13}
and in the language report \cite{Hanus16Curry}.

Curry is a declarative multi-paradigm language
seamlessly combining features from functional and
logic programming.
The syntax of Curry is close to Haskell \cite{PeytonJones03Haskell}.
In addition to Haskell, Curry allows
\emph{free} (\emph{logic}) \emph{variables}
in conditions and right-hand sides of rules.
Thus, \emph{expressions} in Curry programs contain \emph{operations}
(defined functions), \emph{constructors} (introduced in data type declarations),
and \emph{variables} (arguments of operations or free variables).
Function calls with free variables are evaluated by a possibly
non-deterministic instantiation of demanded arguments
\cite{AntoyEchahedHanus00JACM}.
In contrast to Haskell, rules with overlapping left-hand sides
are non-deterministically (rather than sequentially) applied.

\begin{example}\label{ex:insert}
The following example shows the definition of
a non-deterministic list insertion operation in Curry:
\begin{curry}
insert :: a -> [a] -> [a]
insert x ys     = x : ys
insert x (y:ys) = y : insert x ys
\end{curry}
For instance, the expression \code{insert$\;$0$\;$[1,2]}
non-deterministically evaluates to one of the values
\code{[0,1,2]}, \code{[1,0,2]}, or \code{[1,2,0]}.
Based on this operation, we can easily define permutations:
\begin{curry}
perm :: [a] -> [a]
perm []     = []
perm (x:xs) = insert x (perm xs)
\end{curry}
Thus, \code{perm$\;$[1,2,3,4]} non-deterministically evaluates to
all 24 permutations of the input list.
\end{example}
Non-deterministic operations, which are
interpreted as mappings from values into sets of values \cite{GonzalezEtAl99},
are an important feature of contemporary functional logic languages.
Hence, Curry has also a predefined \emph{choice} operation:
\begin{curry}
x ? _  =  x
_ ? y  =  y
\end{curry}
Using non-deterministic operations as arguments could cause
a semantic ambiguity.
Consider the operations:
\begin{curry}
coin = 0 ? 1$\listline$
double x = x + x
\end{curry}
Standard term rewriting produces, among others, the derivation:
\begin{curry}
double coin  $\to~$  coin + coin
             $\to~$  (0 ? 1) + coin
             $\to~$  0 + coin
             $\to~$  0 + (0 ? 1)
             $\to~$  0 + 1
             $\to~$  1
\end{curry}
whose result is (presumably) unintended.
Therefore, Gonz\'alez-Moreno et al.\ \cite{GonzalezEtAl99} proposed
the rewriting logic CRWL as a logical foundation for declarative
programming with non-strict and non-deterministic operations.  This
logic specifies the \emph{call-time choice} semantics \cite{Hussmann92}
where values of the arguments of an operation are set,
though not computed, before the operation is evaluated.
In a lazy strategy, this is naturally obtained
by sharing.
For instance, the two occurrences of \code{coin} in the derivation
above are shared so that \ccode{double coin} has only two results:
\code{0} or \code{2}.
Since standard term rewriting does not conform to
the intended call-time choice semantics,
other notions of rewriting have been proposed to formalize this idea, such as
graph rewriting \cite{EchahedJanodet97GraphRewriting,EchahedJanodet98}
or let rewriting \cite{LopezRodriguezSanchez07}.
In this paper, we use a simple reduction relation,
denoted ``$\to$'',
that we sketch without giving all details (which can be found in
\cite{LopezRodriguezSanchez07}).

In the following, we ignore free (logic) variables
since they can be considered as syntactic sugar
for non-deterministic data generator operations
\cite{AntoyHanus06ICLP}.
Thus, a \emph{value} is an expression without operations or free variables.
To cover non-strict computations, expressions can also
contain the special symbol $\bot$ to represent
\emph{undefined or unevaluated values}.
A \emph{partial value} is a value containing occurrences of $\bot$.
A \emph{partial constructor substitution} is a substitution
that replaces each variable by some partial value.
A \emph{context} $\ctxt{\cdot}$ is an expression with some ``hole''.
Rewrite rules are of the form:
\begin{equation}
  \label{conditional-rule}
  l \mid c \to r
\end{equation}
where ``$\mid c\,$'' is an optional Boolean \emph{condition}.
Such a rule is applicable to an
expression $t$, iff $t$ is an instance of $l$
and the instantiated condition holds.
The rule of (\ref{conditional-rule})
can be replaced by the unconditional rules:
\begin{equation}
  \label{unconditional-rule}  
  \begin{array}{@{}l@{}}
    l \to \emph{ifthen}~ c~ r \\
    \emph{ifthen}~ \code{True}~ x \to x
  \end{array}
\end{equation}
because the values produced by (\ref{conditional-rule})
are all and only those produced by (\ref{unconditional-rule})
\cite{Antoy01PPDP}.
Thus, the reduction relation that we use
throughout this paper is defined as follows:
\begin{center}
\begin{tabular}{@{\quad}l@{\quad}r@{~~}c@{~~}l@{~~~}l}
 \rulename{Fun} &
 $\ctxt{f~\sigma(t_1) \ldots \sigma(t_n)}$ & $\to$ & $\ctxt{\sigma(r)}$ &
where  $f~t_1 \ldots t_n ~\code{=}~ r$ is a program rule\\
 & & & & and $\sigma$ a partial constructor substitution\\[2ex]
\rulename{Bot} & $\ctxt{e}$ & $\to$ & $\ctxt{\bot}$ &
where $e \ne \bot$
\end{tabular}
\end{center}
The first rule models the call-time choice semantics: when a rule is applied,
the actual arguments of the operation must have been evaluated to
partial values. The second rule models the lazy evaluation
(non-strictness) semantics by allowing
the evaluation of any subexpression to an undefined value
(which is intended if the value of this subexpression is not demanded).
As usual, $\tostar$ denotes the reflexive and transitive closure
of this reduction relation.
The equivalence of this rewrite relation and CRWL
is shown in \cite{LopezRodriguezSanchez07}.

\section{Equivalent Operations}
\label{sec:equiv}

As discussed above, equivalence of operations can be defined in
different ways. Ground equivalence and computed result equivalence
only compare the values of applications.
This is too weak since some operations have no finite values.
\begin{example}
\label{ex:ints12}
Consider the following operations that generate infinite lists
of numbers:
\begin{curry}
ints1 n = n : ints1 (n+1)       ints2 n = n : ints2 (n+2)
\end{curry}
Since these operations do not produce finite values,
we cannot detect any difference when comparing only computed results.
However, they behave differently when put into some context,
e.g., an operation that selects the second element of a list:
\begin{curry}
snd (x:y:zs) = y
\end{curry}
Now, \code{snd$\;$(ints1$\;$0)} and \code{snd$\;$(ints2$\;$0)}
evaluate to \code{1} and \code{2}, respectively.
\end{example}
Therefore, we do not consider these operations as equivalent.
This motivates the following notion of equivalence
for possibly non-terminating and non-deterministic operations.\footnote{%
The extension to operations with several arguments is straightforward.
For the sake of simplicity, we formally define our notions
only for unary operations.}
\begin{definition}[Equivalence]
\label{def:equivalence}
Let $f_1,f_2$ be operations of type $\tau \to \tau'$.
$f_1$ is \emph{equivalent} to $f_2$ iff, for any expression $E_1$,
$E_1 \tostar v$ iff $E_2 \tostar v$, where $v$ is a value
and $E_2$ is obtained from $E_1$ by replacing each occurrence of
$f_1$ with $f_2$.
\end{definition}
This notion of equivalence conforms with the usual notion
of contextual equivalence in programming languages
(e.g., see \cite{Pitts00} for a tutorial).
It was already proposed in \cite{AntoyHanus12PADL}
as the notion of equivalence for functional logic programs
and also defined in \cite{BacciEtAl12} as ``contextual equivalence''
for functional logic programs.

Thus, \code{ints1} and \code{ints2} are not equivalent.
Moreover, even terminating operations that always compute
same results might not be equivalent if put into some context.
\begin{example}\label{ex:permsort}
Consider the definition of lists sorted in ascending order:
\begin{curry}
sorted []       = True
sorted [_]      = True
sorted (x:y:zs) = x<=y && sorted (y:zs)
\end{curry}
We can use this definition and the definition of permutations above
to provide a precise specification of sorting a list
by computing some sorted permutation:
\begin{curry}
sort xs | sorted ys = ys    where ys = perm xs
\end{curry}
We might try to obtain an even more compact formulation
by defining the ``sorted'' property as an operation that
is the (partial) identity on sorted lists:
\begin{curry}
idSorted []              = []
idSorted [x]             = [x]
idSorted (x:y:zs) | x<=y = x : idSorted (y:zs)
\end{curry}
Then we can define another operation to sort a list by
composing \code{perm} and \code{idSorted}:
\begin{curry}
sort' xs = idSorted (perm xs)
\end{curry}
Although both \code{sort} and \code{sort'} compute sorted lists,
they might behave differently in some context.
For instance, suppose we want to compute the minimum of
a list by returning the head element of the sorted list:
\begin{curry}
head (x:xs) = x
\end{curry}
Then \code{head$\;$(sort$\;$[3,2,1])} returns \code{1}, as expected,
but \code{head$\;$(sort'$\;$[3,2,1])} returns \code{1} as well as \code{2}.
The latter unintended value is obtained by computing the
permutation \code{[2,3,1]} so that \code{head (idSorted [2,3,1])}
returns \code{2}, since the call to \code{idSorted$\;$[3,1]}
is not evaluated due to non-strictness.
\end{example}
This example shows that our strong notion of equivalence
is reasonable. However, testing this equivalence
might require the generation of arbitrary contexts.
Therefore, we show in the next section how to avoid this context
generation.

\section{Refined Equivalence Criteria}
\label{sec:refining}

The definition of equivalence as stated in
Def.~\ref{def:equivalence} covers the intuition that
equivalent operations can be interchanged at any place in an
expression without changing the expression's value.
Proving such a general form of equivalence could be difficult.
Therefore, we define another form of equivalence that
is based on a single operation that observes the computed results
of the operation under scrutiny.
\begin{definition}[Observable equivalence]
\label{def:obsequivalence}
Let $f_1,f_2$ be operations of type $\tau \to \tau'$.
$f_1$ is \emph{observably equivalent} to $f_2$ iff,
for all operations $g$ of type $\tau' \to \tau''$,
all expressions $e$ and values $v$,
$g~(f_1~e) \tostar v$ iff $g~(f_2~e) \tostar v$.
\end{definition}
We can expect that proving observable equivalence is
easier than equivalence since we trade a context made of an arbitrary expression
that has multiple occurrences of a function $f$ with a single
function call that has a single occurrence of $f$ in a fixed position.
Fortunately, the next theorem shows that proving
observable equivalence is sufficient in general.

We recall some less familiar notations used in the proof of the
next result.  A \emph{path} $p$ in an expression $e$
is a sequence of positive integers that identifies
a \emph{subexpression} of $e$, denoted $e|_p$.
The subexpression is inductively defined as follows.
If $p=[\,]$, i.e., the path is empty, then $e|_p = e$, i.e., $e$ itself.
Otherwise, the following conditions must hold:
$e$ must be of the form $f\,t_1 \ldots t_n$,
for some symbol $f$ of arity $n$ and expressions $t_1, \ldots t_n$,
and $p$ must be of the form $i:is$, where $1 \leq i \leq n$
and $is$ is a path in $t_i$.
In this case, we define $e|_{i:is} = t_i|_{is}$.
We further build of the above concepts.
Let $e$ and $t$ be expressions and $p$ a path in $e$.
The notation $e[t]_p$ denotes an expression, $e'$, that is equal
to $e$ except at position $p$, where $e'|_p = t$.
In other words, $e'$ is obtained from $e$ by replacing its subexpression
at $p$ with $t$.
\begin{theorem}
\label{theo:single-context-equivalence}
Let $f_1,f_2$ be operations of type $\tau \to \tau'$.
$f_1$ and $f_2$ are equivalent iff they are observably equivalent.
\end{theorem}
\begin{proof}
It is trivial that equivalence implies observable equivalence.
Hence, we assume that $f_1$ are $f_2$ are observably
equivalent, i.e., for all operations $g$ of type $\tau' \to \tau''$,
all expressions $e$ and values $v$,
$g~(f_1~e) \tostar v$ iff $g~(f_2~e) \tostar v$.
We show by induction on the number $n$ of occurrences of the symbol $f_1$
the following claim:
\begin{quote}
If $E_1$ is an expression with $n$ occurrences of $f_1$,
$E_2$ is obtained from $E_1$ by replacing each occurrence of
$f_1$ with $f_2$, and $v$ is a value,
then $E_1 \tostar v$ iff $E_2 \tostar v$.
\end{quote}
Base case ($n=0$): Since $E_1$ contains no occurrence of $f_1$,
$E_2 = E_1$ and the claim is trivially satisfied.

Inductive case ($n>0$): Assume the claim holds for $n-1$
and $E_1$ contains $n$ occurrences of $f_1$ and $E_1 \tostar v$ 
for some value $v$. We have to show that $E_2 \tostar v$
(the opposite direction is symmetric) where $E_2$ is obtained
from $E_1$ by replacing each occurrence of $f_1$ with $f_2$.
Let $p$ be a position in $E_1$ with $E_1|_p = f_1~e$ and $e$ does not contain
any occurrence of $f_1$.
Since $E_1 \tostar v$, by definition of $\tostar$,
there is a partial value $t_1$ with $f_1~e \tostar t_1$
and $E_1[t_1]_p \tostar v$.
We define a new operation $g$ by:
\begin{curry}
$g$ $x$ = $E_1[x]_p$
\end{curry}
where $x$ is a new variable that does not occur in $E_1$.
Hence $g~(f_1~e) \tostar g~t_1 \to E_1[t_1]_p \tostar v$.
Our assumption implies $g~(f_2~e) \tostar v$.
By definition of $\tostar$,
there is a partial value $t_2$ with
$g~(f_2~e) \tostar g~t_2 \to E_1[t_2]_p \tostar v$.
Since $E_1[t_2]_p$ contains $n-1$ occurrences of $f_1$,
the induction hypothesis implies that $E_2[t_2]_p \tostar v$.
Therefore, $E_2 = E_2[f_2~e]_p \tostar E_2[t_2]_p \tostar v$.
\end{proof}
A proof that two operations are observably equivalent
could still be difficult since we have to take all
possible functions of a program into account.
However, the next result shows that it is sufficient
to verify that two operations yield always the same partial values
on identical inputs.

\begin{theorem}
\label{theo:equal-partial-value-equivalence}
Let $f_1,f_2$ be operations of type $\tau \to \tau'$.
If, for all expressions $e$ and partial values $t$,
$f_1~e \tostar t$ iff $f_2~e \tostar t$,
then $f_1$ and $f_2$ are equivalent.
\end{theorem}
\begin{proof}
By Theorem~\ref{theo:single-context-equivalence} it is
sufficient to show the observable equivalence of $f_1$ and $f_2$.
Hence, let $g$ be an operation of type $\tau' \to \tau''$,
$e$ an expression and $v$ a value with $g~(f_1~e) \tostar v$.
We have to show that $g~(f_2~e) \tostar v$ (the other direction is symmetric).
By definition of $\tostar$,
there is some partial value $t$ with $f_1~e \tostar t$ and $g~t \tostar v$.
By the assumption of the theorem, $f_2~e \tostar t$.
Hence, $g~(f_2~e) \tostar g~t \tostar v$.
\end{proof}
Note that considering \emph{partial} result values, as opposed to
result values, is essential to establish equivalence.
For instance, the operations \code{sort} and \code{sort'}
defined in Sect.~\ref{sec:equiv} compute the same \emph{value}
for the same input.  But they compute different partial
values for some input.  For example,
$\code{sort' [2,3,1]} \tostar \code{2:}\bot$
whereas \code{sort [2,3,1]} cannot be derived to \code{2:$\bot$}.
If we limit our consideration to result values only,
we could not determine that
\code{sort} and \code{sort'} are not equivalent.

The following result is the converse of
Theorem~\ref{theo:equal-partial-value-equivalence}.
It shows that not only having the same partial values is a sufficient
condition for the equivalence of function, but also a necessary condition.
For partial values $t$ and $u$, we write $t < u$ iff
$t$ is obtained by one or more applications of the \rulename{Bot}
rule to $u$.
It follows that if $u$ is a partial value of an expression $e$,
then any $t < u$ is also a partial value of $e$.
\begin{theorem}
\label{theo:converse-equal-partial-value-equivalence}
Let $f_1,f_2$ be operations of type $\tau \to \tau'$.
If, for some expression $e$, the partial values of
$f_1~e$ differ from those of $f_2~e$,
then $f_1$ and $f_2$ are not equivalent.
\end{theorem}
\begin{proof}
We construct a function $g$ that, under the statement hypothesis,
witnesses the non-equivalence of $f_1$ and $f_2$.
Let $T_1$ be the set of partial values of $f_1~e$
and $T_2$ the set of partial values of $f_2~e$.
W.l.o.g., we assume that there exists some partial value $t \in T_1$
such that $t \not \in T_2$.
Let $g$ be defined by the single rule:
\begin{displaymath}
  g~\bar t \to 0
\end{displaymath}
where $\bar t$ is the expression
obtained from $t$ by replacing
any instance of $\bot$ with a fresh variable.
Then, $g~(f_1~e) \tostar g~t \to 0$,
whereas we show that $g~(f_2~e) \not \tostar 0$.
Suppose the contrary.  Then, it must be that $f_2~e \tostar u$ with
$u$ is an instance of $\bar t$.  This implies $t < u$,
which in turn implies $t \in T_2$.
\end{proof}
The next corollary is useful to avoid the consideration of
all argument expressions in equivalence proofs.
\begin{corollary}
\label{cor:equal-partial-value-equivalence}
Let $f_1,f_2$ be operations of type $\tau \to \tau'$.
If, for all partial values $t$ and $t'$,
$f_1~t \tostar t'$ iff $f_2~t \tostar t'$,
then $f_1$ and $f_2$ are equivalent.
\end{corollary}
\begin{proof}
Assume that $f_1~t \tostar t'$ iff $f_2~t \tostar t'$ holds
for all partial values $t$ and $t'$.
Consider an expression $e$ and a partial value $t_1$
such that $f_1~e \tostar t_1$.
By definition of $\tostar$, there is a partial value $t_0$
with $e \tostar t_0$ and $f_1~t_0 \tostar t_1$.
Our assumption implies $f_2~t_0 \tostar t_1$.
Hence $f_2~e \tostar f_2~t_0 \tostar t_1$.
Since the other direction is symmetric,
Theorem~\ref{theo:equal-partial-value-equivalence}
implies the equivalence of $f_1$ and $f_2$.
\end{proof}
Hence, we have a sufficient criterion for equivalence
checking which does not require the enumeration of arbitrary contexts.
Instead, it is sufficient to test the equivalence on all partial values.
Such a test can be performed by property-based test tools,
as shown in the next section.

One may wonder whether the consideration of values
instead of partial values is enough for equivalence checking.
The next example shows that the answer is negative.
\begin{example}
\label{ex:non-values}
Consider the following operations that take and return Booleans:
\begin{curry}
f1 True  = True                           f2 _ = True
f1 False = True
\end{curry}
Functions \code{f1} and \code{f2} behave identically
on every input \emph{value}.
However, \code{f1$\;\bot$} has no value, whereas \code{f2$\;\bot$} has
value \code{True}.  Thus, values as arguments are not as discriminating
as partial values to expose a difference in behavior, whereas
partial values are as discriminating as expressions.
Actually, \code{f1} and \code{f2} are not equivalent:
consider the operation \code{failed} which has no value.\footnote{%
A possible definition is: \fcode{failed = head []}}
Then \code{f2$\;$failed} has value \code{True} whereas
\code{f1$\;$failed} has no value.
\end{example}
Corollary~\ref{cor:equal-partial-value-equivalence} requires
to compare all partial result values and not just computed results.
The former is more laborious since an expression
might evaluate to many partial values even if it has a single
value. For instance, consider the list generator:
\begin{curry}
fromTo m n = if m>n then [] else m : fromTo (m+1) n
\end{curry}
The expression \code{fromTo 1 5} evaluates to the single value
\code{[1,2,3,4,5]}.
According to the reduction relation defined in
Sect.~\ref{sec:flp},
the same expression reduces to the partial values
\code{$\bot$},
\code{$\bot$:$\bot$},
\code{1:$\bot$},
\code{$\bot$:$\bot$:$\bot$},
\code{1:$\bot$:$\bot$},
\code{$\bot$:2:$\bot$},
\code{1:2:$\bot$},~\ldots~
If operations are non-terminating, it is necessary to consider partial result
values in general.
For instance, \code{ints1$\;$0} and \code{ints2$\;$0}
do not evaluate to a value but they evaluate to the
different partial values \code{0:1:$\bot$} and \code{0:2:$\bot$},
respectively, which shows the non-equivalence of
\code{ints1} and \code{ints2} by
Cor.~\ref{cor:equal-partial-value-equivalence}.
Thus, one may wonder whether for ``well behaved'' operations
it suffices to consider only result \emph{values}.
This would save some effort in the property-based checking
approach described in the next section.
The good behavior could be captured by the property that a
function returns a value for any argument value,
see Def.~\ref{def:well-behaved}.
Unfortunately, the answer is negative as
Example \ref{ex:well-behaved} will show.

\begin{definition}[Terminating, totally defined]
\label{def:well-behaved}
Let $f$ be an operation of type $\tau \to \tau'$.
$f$ is \emph{terminating} iff, for all values $t$ of type $\tau$,
any rewrite sequence $f~t \to t_1 \to t_2 \to\cdots$ is finite.
$f$ is \emph{totally defined} iff, for any value $t$ of type $\tau$,
$f~t$ evaluates to a value $v$ of type $\tau'$.
\end{definition}

\begin{example}
\label{ex:well-behaved}
Functions \code{h1} and \code{h2}, defined below,
are totally defined and terminating.
For any Boolean value $t$, \code{h1\,$t$} and \code{h2\,$t$}
produce the same value result, namely \code{Just\,$t$}.  However, 
\code{h1} and \code{h2} are not observably equivalent when applied to
some partial value, e.g., \code{failed}, as witnessed by \code{g}:
\begin{curry}
h1 True  = Just True                    h2 x = Just x
h1 False = Just False                   g (Just _) = 0
\end{curry}
\end{example}
\removelastskip
This example shows that we have to use partial input values
for equivalence tests
even when all the involved operations are terminating and totally defined.
This requirement was already suggested by Example~\ref{ex:non-values}
in which the operations were terminating and totally defined.

The following discussion further investigate a condition of
good behavior that affects the equivalence of two functions.
Since the condition is strong, its practical applicability is limited.
However, this result sheds light on the relationships
between equivalence and other properties of functions.
\begin{definition}
  We say that a function $f$ is \emph{deterministically defined}
  iff any redex rooted by $f$ has only one reduct.
\end{definition}
For example, operation \code{insert} of Example \ref{ex:insert} is
not deterministically defined.
The redex \code{insert\,0\,[1]} has two distinct reducts,
\code{[0,1]} and \code{[1:insert\,0\,[]]}.
By contrast, both operations \code{f1} and \code{f2}
of Example \ref{ex:non-values} are deterministically defined.

The next lemma investigates properties of a program
in which every operation is both totally
and deterministically defined.
Therefore, we must exclude the $\bot$ symbol
from the signature and rule \rulename{Bot} from the program,
since they would violate the assumptions.
In fact, if we reduce a \rulename{Fun}-rule redex with
the \rulename{Bot} rule, the same redex has two reducts,
hence we lose the determinism of definitions.
Likewise, if we allow the $\bot$ symbol, then
the $\bot$ expression is neither a value nor can be reduced to a value,
hence we lose the totality of definitions. 
Eliminating these elements from the following discussion
is not a problem.
Both the $\bot$ symbol and the \rulename{Bot} rule
are convenient abstractions to \emph{reason} about non-strict computations
by ignoring subexpressions,
but the next lemma shows that this is not necessary
for totally and deterministically defined operations
since their evaluation always ends in a unique value.
\begin{lemma}
  \label{unique-and-finite}
  Let $P$ be a program in which every operation is both totally
  and deterministically defined.
  Then, for every expression $e$ over the signature of $P$:
  (1) there exists a unique value $v$ such that $e \tostar v$,
  and (2) there exists no infinite computation of $e$.
\end{lemma}
\begin{proof}
  The existence of $v$ is proved by induction on the number $n$ of
  occurrences of operation symbols in $e$.  For the base case, $n=0$,
  the claim is   witnessed by $v=e$.
  For the inductive case, let $t$ be an outermost
  operation-rooted subexpression of $e$.  By the induction hypothesis,
  every operation-rooted proper subexpression of $t$ is reducible to a value.
  Hence, there exists an expression $t'$ such that $t \tostar t'$ and the
  only occurrence of an operation in $t'$ is the root.
  By the assumption of total definition, there exists some value $v_t$
  such that $t' \tostar v_t$. Hence, there exists some value $v$ such that
  $e \tostar v$.

  We now prove the uniqueness of $v$.  Observe that program $P$ is almost
  orthogonal \cite[p.52]{Ohlebusch02AdvancedTR}.
  In fact, the assumption that $P$ is a functional logic program
  implies that $P$ follows the constructor discipline, hence
  rule left-hand sides can overlap only at the root,
  and the assumption that functions
  are deterministically defined implies that critical pairs,
  if any, are trivial.
  Since $P$ is almost orthogonal, the Parallel Moves Lemma holds for
  the computations of $P$ \cite[p. 56]{Ohlebusch02AdvancedTR}.
  Thus, let $e \tostar v$ be a computation of $e$, where $v$ is some value, 
  and let $e=e_0 \to e_1 \to \ldots e_n = v'$ be some other computation of
  $e$, where $v'$ is some value as well.  By parallel moves, for any
  $i=0,\ldots n$, there is a computation $e_i \tostar v$, in
  particular $e_n \tostar v$.  But since $e_n = v'$ is a value, it
  must be $v'=v$.

  We now prove the termination of $P$.  Given an
  expression $t$, we have shown that there exists a value $v_t$
  and a computation $A_t: t \tostar v_t$, where $A_t$ is an arbitrary
  identifier of the computation.  We denote $size(t)$ the number
  of symbol occurrences of $t$ and $length(A_t)$ the number of steps of $A_t$
  with the stipulation that multiple steps, if any, to reduce distinct
  residuals of some redex contribute only a single step to the length.
  We define a well-founded order, ``$<$'', on the
  expressions over the signature of $P$.
  Given expressions $t$ and $u$, we define $t < u$ iff
  $size(t) < size(u)$ or $size(t) = size(u)$ and $length(A_t) < length(A_u)$.
  Let $e$ be any expression over the signature of $P$.
  W.l.o.g.~we assume that $e$ is operation-rooted.
  We have shown that there exists a value $v_e$ such that $A_e : e \tostar v_e$.
  We show by induction on ``$<$'' that any computation of $e$ is finite.
  By contradiction, suppose there exists a non-terminating computation
  $B: e = e_0 \to e_1 \to \ldots$~
  The case when $length(A_e) = 0$ is immediate,
  since $e$ is a value and $B$ cannot exist.
  If, for some $i>0$, the step $e_{i-1} \to e_i$ contracts a redex
  contracted in $A_e$, then by parallel moves,
  there is a computation $e_i \tostar v_e$ shorter than $length(A_e)$
  and by the induction hypothesis any computation of $e_i$ is finite,
  which contradicts the existence of $B$.
  Thus, all the steps of $B$ must be within some proper subexpression, say $t$,
  of $e$.  Since $size(t) < size(e)$,
  by the induction hypothesis any computation of $t$ is finite,
  which again contradicts the existence of $B$.
\end{proof}
\begin{theorem}
  \label{ground-implies}
  Let $P$ be a program in which every operation is both totally
  and deterministically defined,
  and let $f_1,f_2$ be operations of $P$ of type $\tau \to \tau'$.
  If $f_1$ and $f_2$ are ground equivalent,
  then $f_1$ and $f_2$ are equivalent.
\end{theorem}
\begin{proof}
  By Lemma \ref {unique-and-finite},
  for every expression $e$ over the signature of $P$
  there exists a unique value $v$ such that $e \tostar v$
  and no computation of $e$ is infinite.
  The proof is by induction on the number of occurrences of
  operation symbols in $e$.
  We prove the observable equivalence of $f_1$ and $f_2$.
  Let $g$ be any function of $P$ of type $\tau' \to \tau''$
  and $e$ any expressions of type $\tau'$ over the signature of $P$.
  Suppose that for some value $v$, $g\,(f_1\,e) \tostar v$.
  There exist a unique value $t$ such that
  $f_1\,e \tostar t$ and $g\,t \tostar v$.
  By the assumption of ground equivalence of $f_1$ and $f_2$,
  $f_2\,e \tostar t$ as well, thus $g\,(f_2\,e) \tostar v$.  
\end{proof}
As already mentioned, the requirement for total and deterministic
operations is strong but not unusual for some kinds of languages.
For instance, the language Agda \cite{Norell08}
is a functional language with dependent types intended
to develop verified programs \cite{Stump16}.
To ensure a consistent logic, all functions must be terminating
and deterministically and totally defined (the latter requirement
can be relaxed by requiring proofs that undefined cases
cannot occur).
Thus, checking ground equivalence is sufficient for Agda programs.

Now we have enough refined criteria to implement
an equivalence checker with a property-based checking tool.

\section{Property-based Checking}
\label{sec:currycheck}

Property-based testing is a useful technique to obtain
reliable software systems.
Testing cannot verify the correctness of programs,
but it can be performed automatically and it might prevent
wasting time when attempting to prove incorrect properties.
If proof obligations are expressed as properties,
i.e., Boolean expressions parameterized over input data,
and we test these properties with a large set of input data,
we have a higher confidence in the correctness of the properties.
This motivates the use of property testing tools
which automate the checking of properties
by random or systematic generation of test inputs.
Property-based testing has been introduced with the
QuickCheck tool \cite{ClaessenHughes00} for the functional language Haskell
and adapted to other languages, like
PrologCheck \cite{AmaralFloridoSantosCosta14} for Prolog,
PropEr \cite{PapadakisSagonas11} for the
concurrent functional language Erlang,
and EasyCheck \cite{ChristiansenFischer08FLOPS}
and CurryCheck \cite{Hanus16LOPSTR} for the
functional logic language Curry.
If the test data is generated in a systematic (and not random) manner,
like in SmallCheck \cite{RuncimanNaylorLindblad08},
GAST \cite{KoopmanAlimarineTretmansPlasmeijer03},
EasyCheck \cite{ChristiansenFischer08FLOPS},
or CurryCheck \cite{Hanus16LOPSTR},
these tools can actually verify properties for finite input domains.
In the following, we show how to extend the property-based test tool
CurryCheck to support equivalence checking of operations.

\subsection{Equivalence Testing with CurryCheck}

Properties can be defined in source programs
as top-level entities with result type \code{Prop}
and an arbitrary number of parameters.
CurryCheck offers a predefined set of property combinators
to define properties.
In order to compare expressions involving
non-deterministic operations,
CurryCheck offers the property \ccode{<\char126>}
which has the type \code{a$\,\to\,$a$\,\to\,$Prop}.\footnote{%
Actually, the parameter type \code{a} must also support
the type contexts \code{Eq} to compare values and
\code{Show} to show test inputs.}
It is satisfied if both arguments have identical result sets.
For instance, we can state the requirement that
permutations do not change the list length by the property:
\begin{curry}
permLength xs = length (perm xs) <~> length xs
\end{curry}
Since the left argument of \ccode{<\char126>}
evaluates to many (expectedly identical) values,
it is relevant that \ccode{<\char126>} compares result \emph{sets}
(rather than multi-sets).
This is reasonable from a declarative programming point of view,
since it is irrelevant how often some result is computed.

Corollary~\ref{cor:equal-partial-value-equivalence} provides a specific
criterion for equivalence testing:
Two operations \code{f1} and \code{f2} are equivalent
if, for any partial argument value, they produce
the same partial result value.
Since partial values cannot be directly compared,
we model partial values by extending total values with
an explicit ``bottom'' constructor representing the partial value $\bot$.
For instance, consider the data types used
in Example~\ref{contrived-ground-equiv}.
Assume that they are defined by:
\begin{curry}
data AB = A | B$\listline$
data C  = C AB
\end{curry}
We define their extension to partial values by renaming all constructors
and adding a bottom constructor to each type:
\begin{curry}
data P_AB = Bot_AB | P_A | P_B$\listline$
data P_C  = Bot_C  | P_C P_AB
\end{curry}
In order to compare the partial results of two operations,
we introduce operations that return 
the partial value of an expression w.r.t.\ a given
partial value, i.e., the expression is partially evaluated up to the degree
required by the partial value (and it fails if the expression
has not this value). These operations can easily be implemented
for each data type:
\begin{curry}
peval_AB :: AB -> P_AB -> P_AB
peval_AB _ Bot_AB = Bot_AB                 -- no evaluation
peval_AB A P_A    = P_A
peval_AB B P_B    = P_B$\listline$
peval_C :: C -> P_C -> P_C
peval_C _     Bot_C   = Bot_C              -- no evaluation
peval_C (C x) (P_C y) = P_C (peval_AB x y)
\end{curry}
We explain the rules of \code{peval\_C}.
The first rule does not require any evaluation of its first
argument, of type \code{C},
a condition indicated by \code{Bot\_C} in the second argument.
Therefore, the rule returns \code{Bot\_C} that stands for a
totally unevaluated expression of type \code{C}.
The second rule handles the case in which 
the first argument must be evaluated to a head constructor form,
a condition indicated by \code{P\_C} as the root of the second argument.
Hence, the first argument must be rooted by \code{C}.
The arguments \code{x} and \code{y} of
\code{C} and \code{P\_C}, respectively, must be recursively matched.  
This is accomplished by the arguments of \code{P\_C} in the
right-hand side of the rule.

Now we can test the equivalence of \code{f} and \code{g},
defined in Example~\ref{contrived-ground-equiv},
by evaluating both operations to the same partial value.
Thus, a single test consists of the application of each
operation to an input \code{x} and a partial result value \code{p}
together with checking
whether these applications produce \code{p}:
\begin{curry}
f_equiv_g :: C -> P_C -> Prop
f_equiv_g x p = peval_C (f x) p <~> peval_C (g x) p
\end{curry}
This property is checked over a set of partial values for \code{x}.
These values are generated by CurryCheck using a sophisticated
algorithm that attempts to maximize coverage through narrowing
and also uses a random component.
Finite, small domains, like \code{AB} and \code{C},
are exhaustively explored, thus
CurryCheck generates, among others, the inputs
\code{x=failed} and \code{p=(P\us{}C Bot\us{}AB)}
for which the property does not hold.
This shows that \code{f} and \code{g} are not equivalent.

In a similar way, we can model partial list result values
and test whether \code{sort} and \code{sort'},
as defined in Example~\ref{ex:permsort},
are equivalent. If the domain of list elements has
three values (like the standard type \code{Ordering}
with values \code{LT}, \code{EQ}, and \code{GT}),
CurryCheck reports a counter-example (a list with three different
elements computed up to the first element) at the 89th test.
The high number of tests is due to the fact that test inputs
as well as partial output values are enumerated to test each property.

\subsection{Reducing Test Cases}

The number of test cases can be significantly reduced by a different encoding.
Instead of enumerating operation inputs as well as partial result values,
we can enumerate operation inputs only and use a non-deterministic
operation which returns \emph{all} partial result values
of some given expression.
For our example types, these operations can be defined as follows:
\begin{curry}
pvalOf_AB :: AB -> P_AB
pvalOf_AB _ = Bot_AB
pvalOf_AB A = P_A
pvalOf_AB B = P_B$\listline$
pvalOf_C :: C -> P_C
pvalOf_C _     = Bot_C
pvalOf_C (C x) = P_C (pvalOf_AB x)
\end{curry}
Now we can test the equivalence of \code{f} and \code{g} by
checking whether both operations have the same set of partial
values for a given input:
\begin{curry}
f_equiv_g :: C -> Prop
f_equiv_g x = pvalOf_C (f x) <~> pvalOf_C (g x)
\end{curry}
CurryCheck returns the same counter-example as before.
This is also true for the permutation sort example,
but now the counter-example is found at the 11th test.

Due to the reduced search space of our second implementation
of equivalence checking, we might think that this method
should always be preferred.
However, in case of non-terminating operations, it is less powerful.
For instance, consider the operations \code{ints1} and \code{ints2}
of Example~\ref{ex:ints12}.
Since \code{ints1$\;$0} has an infinite set of partial result
values, the equivalence test with \code{pvalOf} operations
would try to compare sets with infinitely many values.
Thus, it would not terminate and would not yield a counter-example.
However, the equivalence test with \code{peval} operations
returns a counter-example by fixing a partial term
(e.g., a partial list with at least two elements)
and evaluating \code{ints1} and \code{ints2} up to this partial list.

Based on these considerations, equivalence checking is implemented
in CurryCheck as follows.
First, CurryCheck provides a specific ``operation equivalence'' property
denoted by \code{<=>}. Hence:
\begin{curry}
f_equiv_g = f <=> g
\end{curry}
denotes the property that \code{f} and \code{g} are equivalent operations.
In contrast to other properties like \ccode{<\char126>},
which are implemented by some Curry code \cite{ChristiansenFischer08FLOPS},
the property \ccode{<=>} is just a marker\footnote{CurryCheck
also ensures that both arguments of \ccode{<=>} are defined operations,
otherwise an error is reported.}
which will be transformed
by CurryCheck into a standard property based on the results
of Sect.~\ref{sec:refining}.
For this purpose, CurryCheck transforms the property above as follows:
\begin{enumerate}
\item
In the general case, CurryCheck tests whether, for each partial value,
the functions \code{f} and \code{g} compute this result.
Thus, if \code{T} is the result type of \code{f} and \code{g},
the auxiliary operation \code{peval\us{}T} (and similarly for
all types on which \code{T} depends) is generated as shown above and
the following property is generated:
\begin{curry}
f_equiv_g x p = peval_T (f x) p <~> peval_T (g x) p
\end{curry}
\item
If operations \code{f} and \code{g} are both terminating,
then the sets of partial result values are finite so that
these sets can be compared in a finite amount of time.
Thus, if \code{T} is the result type of \code{f} and \code{g},
the auxiliary operation \code{pvalOf\us{}T} (and similarly for
all types on which \code{T} depends) is generated as shown above and
the following property is generated:
\begin{curry}
f_equiv_g x = pvalOf_T (f x) <~> pvalOf_T (g x)
\end{curry}
\end{enumerate}

\subsection{Using Termination Information}

In order to decide between these transformation options,
our extension of CurryCheck
uses the analysis framework CASS \cite{HanusSkrlac14}
to approximate the termination behavior of both operations.
If the termination property of both operations can be proved
(for this purpose, CASS uses an ordering on arguments in recursive calls),
the second transformation is used, otherwise the first one is.
If the termination cannot be proved but the programmer is sure
about the termination of both operations, she can also
mark the property with the suffix \code{'TERMINATE}
to tell CurryCheck to use the second transformation.

\begin{example}
Consider the recursive and non-recursive definition of
the McCarthy 91 function:
\begin{curry}
mc91r n = if n > 100  then n-10  else mc91r (mc91r (n+11))$\listline$
mc91n n = if n > 100  then n-10  else 91
\end{curry}
Since CASS is not able to determine the termination of \code{mc91r},
we annotate the equivalence property so that CurryCheck uses
the second transformation:
\begin{curry}
mc91r_equiv_mc91n'TERMINATE = mc91r <=> mc91n
\end{curry}
\end{example}

\subsection{Generating Partial Values}

Due to the results of Sect.~\ref{sec:refining},
the generated properties must be checked with all \emph{partial} input values.
For instance, to check the property:
\begin{curry}
f_equiv_g x = pvalOf_C (f x) <~> pvalOf_C (g x)
\end{curry}
shown above, we have to test it with all partial values
of type \code{AB} for the argument \code{x},
i.e., $\bot$, \code{A}, and \code{B}.
However, CurryCheck generates only total values
for input parameters of properties.
Instead of defining specific generators for partial values
(note that CurryCheck also supports the definition of user-defined generators
for input parameters, as described in \cite{Hanus16LOPSTR}),
we exploit the already available partial data type \code{P\us{}AB}.
This type contains an explicit representation of an undefined element
but, due to typing reasons, \code{P\us{}AB} values
cannot be used as inputs to \code{f} and \code{g}.
However, we can map these values into the desired \code{AB} values:
\begin{curry}
from_P_AB :: P_AB -> AB
from_P_AB Bot_AB = failed
from_P_AB P_A    = A
from_P_AB P_B    = B
\end{curry}
Thus, the explicit bottom element is mapped into a failure,
and all other constructors are mapped into their corresponding original 
constructors.
Now we can modify the property so that
the enumeration of values for the partial data types
results in the desired partial inputs to \code{f} and \code{g}:
\begin{curry}
f_equiv_g :: P_C -> Prop
f_equiv_g px = let x = from_P_AB px
               in pvalOf_C (f x) <~> pvalOf_C (g x)
\end{curry}
The definition of equivalence properties with \code{P\us{}AB}
input values has also the advantage that we can define
a reasonable string representation to show input values
for possible counter-examples.
Since CurryCheck uses the predefined operation \code{show},
defined in the type class \code{Show},\footnote{Note that
recent Curry implementations also support type classes as in Haskell.}
to produce a string representation of counter-examples,
CurryCheck generates the following \code{Show} instance
for \code{P\us{}AB}:
\begin{curry}
instance Show P_AB where
  show Bot_AB = "failed"
  show P_A    = "A"
  show P_B    = "B"
\end{curry}
Such instances are generated for all partial data types
involved in properties.
This is useful to report in a user-friendly format
counter examples to the supposed equivalence of operations.
For example, if we apply CurryCheck to test the equivalence of
\code{sort} and \code{sort'}, as defined in Example~\ref{ex:permsort},
CurryCheck reports the following counter-example:
\begin{curry}
Arguments:
(1 : (0 : failed))
(failed : failed)
\end{curry}
The first argument is the input to the sort operations,
and the second argument is the partial result value
which must be evaluated to expose the non-equivalence
of \code{sort} and \code{sort'}.

\subsection{Productive Operations}

According to the results of Sect.~\ref{sec:refining},
checking the above properties allows us to find counter-examples
for non-equivalent operations if the domain of values
is finite (as in the example of Sect.~\ref{sec:intro})
or we enumerate enough test inputs.
An exception are specific non-terminating operations.

\begin{example}
Consider the contrived operations:
\begin{curry}
k1 = [loop,True]
k2 = [loop,False]
\end{curry}
where the evaluation of \code{loop} does not terminate.
The non-equivalence of \code{k1} and \code{k2}
can be detected by evaluating them to
\code{[$\bot$,True]} and \code{[$\bot$,False]}, respectively.
Since a systematic enumeration of all partial values
might generate the value \code{[True,$\bot$]}
before \code{[$\bot$,True]},
CurryCheck might not find the counter-example due to the non-termination
of \code{loop} (since CurryCheck performs all tests in a sequential manner).
\end{example}
Fortunately, this is a problem which rarely occurs in practice.
Not all non-terminating operations are affected by this problem
but only operations that loop without producing any data.
For instance, the non-equivalence of \code{ints1} and \code{ints2}
of Example~\ref{ex:ints12} can be shown with our approach.
Such operations are called \emph{productive} in \cite{Hanus17ICLP}.
Intuitively, productive operations always generate
some data after a finite number of steps.

In order to avoid such non-termination problems when
CurryCheck is used in an automatic manner (e.g., by a software
package manager), CurryCheck has an option for a
``safe'' execution mode. In this mode, operations
involved in an equivalence property are analyzed for
their productivity behavior. If it cannot be proved that
an operation is productive (by approximating their run-time behavior
with CASS), the equivalence check for this operation
is ignored. This ensures the termination of all equivalence tests.
The restriction to productive operations is not a serious limitation since,
as evaluated in \cite{Hanus17ICLP}, most operations
occurring in practical programs are actually productive.
If there are operations where CurryCheck cannot prove productivity,
but the programmer is sure about this property,
the property can be annotated with the suffix \code{'PRODUCTIVE}
so that it is also checked in the safe mode.

\begin{example}
\label{ex:primes}
Consider the definition of all prime numbers by
the sieve of Eratosthenes:
\begin{curry}
primes = sieve [2..]
 where sieve (x:xs) = x : sieve (filter (\y -> y `mod` x > 0) xs)
\end{curry}
After looking at the first four values of this list, a naive programmer
might think that the following prime generator is much simpler:
\begin{curry}
dummy_primes = 2 : [3,5..]
\end{curry}
Testing the equivalence of these two operations
is not possible in the safe mode, since the productivity of
\code{primes} depends on the fact that there are infinitely many prime
numbers. Hence, a more experienced programmer would annotate the
equivalence test as:
\begin{curry}
primes_equiv'PRODUCTIVE = primes <=> dummy_primes
\end{curry}
so that the equivalence will be tested even in the safe mode
and CurryCheck finds a counter-example (evaluating the result list
up to the first five elements) to this property.
\end{example}

\section{Evaluation}
\label{sec:evaluation}

This section shows a few practical results of our equivalence checking
technique by applying CurryCheck to some examples.
In particular, we will evaluate the influence of optimized
checking for terminating operations on the number of test cases.

Practical results heavily depend on the chosen set of benchmark
programs, as discussed in the following example.

\begin{example}
One can easily construct examples
of non-equivalent operations where CurryCheck will not show
a counter-example. For instance, if a ``loop'' should be compared
with a terminating operation:
\begin{curry}
l1 :: Int                      l2 :: Int
l1 = l1                        l2 = 42
\end{curry}
then the property \code{l1 <=> l2}
will either not be checked (when CurryCheck is executed
in the ``safe'' mode) or,
if it is checked, the check will not terminate.
One can also construct examples where, in principle,
CurryCheck can find counter-examples but they are not found in practice
since the set of tested partial values is too large:
\begin{curry}
m1 :: [Int]                    m2 :: [Int]
m1 = [1..1000000] ++ [1]       m2 = [1..1000000] ++ [2]
\end{curry}
Here, the counter-example must be a list of more than one million elements.
\end{example}
In order to avoid such artificial examples,
we evaluate the behavior of CurryCheck on examples
already discussed in this paper or known from the literature.
Table~\ref{table:benchmarks} summarizes our results.
For each example,\footnote{%
The source code of all examples are available from the
Curry package \fcode{currycheck}.}
it shows the number of tests to find
a counter-example to the stated equivalence property.
The column ``general'' shows this number for the general
transformation, and ``with term.'' shows the number
for the improved transformation when the termination of operations
is taken into account (or ``n/a'' if the improvement
is not applicable since the operations are not terminating).
The rightmost column contains ``yes'' or ``$\infty$''
depending on whether the check for the corresponding ground equivalence property
succeeds or does not terminate, respectively.
The entries in this column demonstrate that checking
ground equivalence only is not sufficient.

\begin{table}[t]
\begin{center}
\begin{tabular}{|l|c|c|c|}
\hline
Example & general & with term. & ground equiv. \\
\hline
\cite{BacciEtAl12} (Ex.~\ref{contrived-ground-equiv}) &    2 &   1 & yes \\
\code{Intersperse} \cite{ChristiansenSeidel11}        &   43 &   4 & yes \\
\code{Ints12}      (Ex.~\ref{ex:ints12})              &   47 & n/a & $\infty$ \\
\code{MultBin}     \cite{Christiansen11}              & 1041 &  42 & yes \\
\code{MultPeano}   \cite{Christiansen11}              &   24 &   9 & yes \\
\code{NDInsert} (Ex.~\ref{ex:insert}, \cite{Hanus13}) &    7 &   1 & yes \\
\code{Perm}     (Ex.~\ref{ex:insert}, \cite{Hanus13}) &   13 &   3 & yes \\
\code{Primes}      (Ex.~\ref{ex:primes})              &   38 & n/a & $\infty$ \\
\code{RevRev}                                         &   13 &   3 & yes \\
\code{SortEquiv}   (Ex.~\ref{ex:permsort})            &   89 &  11 & yes \\
\code{SortPermute} (Ex.~\ref{ex:permsort}, \cite{Hanus13}) & 1174 &  46 & yes \\
\code{Take}        \cite{FonerZhangLampropoulos18}    &   11 &   2 & yes \\
\code{Unzip}       \cite{Chitil11TR}                  &   27 &  11 & yes \\
\hline
\end{tabular}
\end{center}
\caption{Results of CurryCheck to disprove equivalences.
  Each example consists of two functions whose equivalence,
  or lack thereof, is investigated by CurryCheck.
  Column ``general'' reports the number of tests generated
  before a counterexample to the equivalence is found.
  Column ``with term.'' report the same value under the
  assumption of termination.  
}
\label{table:benchmarks}
\end{table}

Several of these examples are already discussed in this paper
(see examples~\ref{contrived-ground-equiv}, \ref{ex:ints12},
\ref{ex:permsort}, and \ref{ex:primes}).
A definition of a non-deterministic list insertion operation
by exploiting rules with overlapping left-hand sides
has been shown in Example~\ref{ex:insert}.
Sometimes, an alternative formulation with
non-overlapping left-hand sides and the choice operation in the
right-hand side is used \cite{Hanus13}:
\begin{curry}
insert' :: a -> [a] -> [a]
insert' e []     = [e]
insert' e (x:xs) = (e : x : xs) ? (x : insert' e xs)
\end{curry}
Although \code{insert} and \code{insert'} are ground equivalent,
they are not equivalent:
the expression
\code{head$\;$(insert$\;$1$\;$failed)} evaluates to \code{1}
whereas no value is computed when \code{insert} is replaced
by \code{insert'}.
This non-equivalence is easily detected by CurryCheck
(see example \code{NDInsert}).

The non-equivalence of \code{insert} and \code{insert'}
results also in the non-equivalence of the permutation operations
if they are defined as shown in Example~\ref{ex:insert}.
This is shown with \code{Perm} in Table~\ref{table:benchmarks}.
This slight difference in the definition of the insert operation
might require much more tests to disprove the equivalence
of sort operations based on them.
Example \code{SortPermute} in Table~\ref{table:benchmarks}
is similar to Example~\ref{ex:permsort} but use \code{insert'}
instead of \code{insert}.

The remaining examples are purely functional programs.
\code{Intersperse} compares two versions of
the list operation \code{intersperse} which inserts an
element between all succeeding elements of a list.
These two versions are discussed in \cite{ChristiansenSeidel11}
w.r.t.\ their strictness behavior (see also below).
Similarly, \code{MultPeano} and \code{MultBin} are definitions
of multiplication on Peano numbers and a binary representation
of natural numbers, respectively, which are used in
\cite{Christiansen11} to evaluate the tool Sloth for the
strictness analysis of functions.
\code{Take} compares two versions of the prelude operation \code{take}
which returns a finite prefix from a given list
\cite{FonerZhangLampropoulos18}.
The example \code{Unzip} compares two versions of the prelude operation
\code{unzip} which transforms a list of pairs into a pair of lists.
These definitions are used in \cite{Chitil11TR}
as examples for different strictness behavior.
Finally, \code{RevRev} shows the non-equivalence of a double list reverse
and the identity operation.

The numbers shown in Table~\ref{table:benchmarks}
indicate that it is useful to take the termination behavior
of operations into account.
For instance, CurryCheck typically performs only a few hundred tests
for each property.
If the number of tests is not increased from its standard value,
the non-equivalences in examples \code{MultBin} or \code{SortPermute}
would not be detected with the general transformation scheme.

\label{sec:strictness}
An observation from these examples is that seemingly equivalent
operations are often not equivalent because
they demand a different degree of their inputs
in order to compute some result.
This kind of demand is called \emph{strictness} and analyzing the
strictness properties of operations has a long tradition
in functional programming \cite{Mycroft80}.
Strictness can also be extended to functional logic languages
where one has to take into account that functional logic programs
might compute different results on a given input expression.

\begin{definition}[Strictness \cite{Hanus12ICLP}]
\label{def:strictn}
Let $f$ be an operation of type $\tau \to \tau'$.
$f$ is called \emph{strict} or \emph{demands its argument}
iff $v = \bot$ whenever $f~\bot \tostar v$ for some partial value $v$.
\end{definition}
Intuitively, $\bot$ is the only result when evaluating
a strict operation applied to an undefined argument.
The extension to operations with more than one argument is
straightforward.
Furthermore, one can also define more refined notions of strictness,
like spine strictness (demanding the complete
evaluation of a list structure but not the list elements,
e.g., as in the operation \code{length}).

Many examples of non-equivalent operations in this paper
have different strictness properties.
For instance,
\code{f} is not strict but \code{g} is strict
(Sect.~\ref{sec:intro}),
\code{f1} is strict but \code{f2} is not strict
(Example~\ref{ex:non-values}), and
\code{h1} is strict but \code{h2} is not strict
(Example~\ref{ex:well-behaved}).
Thus, we could also take strictness into account
when checking equivalence of operations.
This is justified by the following fact.

\begin{proposition}
\label{prop:strictness-equivalence}
Let $f_1,f_2$ be operations of type $\tau \to \tau'$
where $\tau$ is sensible, i.e., $\tau$ has at least one value.
If $f_1$ is strict and $f_2$ is not strict,
then $f_1$ and $f_2$ are not equivalent.
\end{proposition}
\begin{proof}
Assume that $f_1$ is strict and $f_2$ is not strict.
By Def.~\ref{def:strictn} and the non-strictness of $f_2$
$f_2~\bot \tostar v$ for some partial value $v \neq\bot$.
By Def.~\ref{def:strictn} and the strictness of $f_1$
$f_1~\bot \tostar \bot$ but $f_1~\bot \not\tostar v$.
By Theorem~\ref{theo:converse-equal-partial-value-equivalence},
$f_1$ and $f_2$ are not equivalent.
\end{proof}
Thus, strictness information could be exploited
to improve equivalence testing of two operations $f_1$ and $f_2$ as follows.
\begin{enumerate}
\item
If both $f_1$ and $f_2$ are strict, it is not necessary
to test $f_1$ and $f_2$ on \code{failed} as an input argument.
In particular, for data types containing only constants
(0-ary constructors), the generation of partial input values,
as discussed in Sect.~\ref{sec:currycheck}, can be avoided.
\item
If $f_1$ is strict and the strictness property of $f_2$
is not definitely known,\footnote{Note that
typical strictness analysis tools \cite{Hanus12ICLP}
approximate strictness but not non-strictness, i.e.,
they either return that an operation is strict or
that an operation \emph{might} be non-strict.}
one can first evaluate \code{$f_2\;$failed}.
If it can be evaluated to some constructor-rooted expression,
then $f_1$ and $f_2$ are not equivalent, otherwise one can
proceed with equivalence checking as described
in Sect.~\ref{sec:currycheck}.
\end{enumerate}
Considering Prop.~\ref{prop:strictness-equivalence} and some examples,
one might think that equivalence checking can be split into
two separate parts: ground equivalence checking and strictness analysis.
However, this is not the case since there are ground equivalent
operations with identical strictness properties which are
not equivalent.

\begin{example}
Consider the operations \code{g1} and \code{g2}
defined by the following rules:
\begin{curry}
g1 x = 1 : head []$\listline$
g2 x = 2 : head []
\end{curry}
Since neither \code{g1} nor \code{g2} can be evaluated to some value,
\code{g1} and \code{g2} are trivially ground equivalent.
Obviously, they have identical strictness properties.
However, they are not equivalent since
\code{head$\;$(g1$\;$0)} and \code{head$\;$(g2$\;$0)} evaluate
to \code{1} and \code{2}, respectively.
\end{example}

\section{Applications}
\label{sec:applications}

In this section we describe two applications of our equivalence
checking techniques.
Since our checking approach is fully automatic,
it can be integrated into programming tools for Curry.
In the following, we describe the integration of equivalence
checking in two such tools.

\subsection{Equivalence Checking in a Software Package Manager}
\label{sec:cpm}

As discussed in the introduction,
one motivation of this work is to provide support
for semantic versioning checking in a software package manager.
Software package managers are important tools
for software development.
They use version numbers to identify different versions of a package,
where a meaning is associated with these version numbers,
often following the idea of the semantic versioning standard.
Although these version numbers are used to install
appropriate versions of packages,
there are almost no tools to support the programmer
in checking whether a given version number is correct
w.r.t.\ the semantic versioning scheme.
An exception is the Elm package manager\footnote{\url{http://elm-lang.org}}
which uses the names and types of the API to decide
about appropriate version numbers.
Obviously, such a purely syntactic check cannot detect
semantic differences when API types are not changed,
e.g., when an addition operation is replaced by a multiplication.
To detect such kinds of semantic changes,
one has to verify or at least to test the different versions
of a software package.

The integration of equivalence testing for semantic versioning
was proposed in \cite{Hanus17ICLP}, where property testing
is integrated in the Curry package manager
CPM.\footnote{\url{http://curry-language.org/tools/cpm}}
However, only ground equivalence is tested and,
as we have seen, this is too weak in the context
software packages.
Based on the results developed in this paper,
we have integrated full equivalence checking in CPM.
In the following, we sketch this implementation.

Similarly to many other software package managers,
CPM has a collection of commands to search for packages,
install and upgrade packages by resolving dependency constraints, etc.
The most interesting command is \code{diff} which
compares two versions of a package.
For instance, within the scope of a package,
we can compare the current package to a previous version,
say \code{2.1.4}, of the same package by invoking the command:
\begin{curry}
> cypm diff 2.1.4
\end{curry}
This starts a comparison process between these packages
with the following main steps:
\begin{description}
\item [API checking:]
The signatures of all data types and operations of the API,
i.e., public entities of the exported modules of the package,
are compared.
If there are any syntactic differences in entities
occurring in both packages
and the major version numbers of the packages are identical,
a violation of semantic versioning is reported.
A violation is also reported 
if there is some API entity $f$ occurring in version $a.b_1.c_1$
but not in version $a.b_2.c_2$ and $b_1$ is not greater than $b_2$.
\item[Behavior equivalence testing:]
If the major version numbers of the packages under comparison are identical,
then, for each API operation occurring in both package versions,
the equivalence of both versions of the operation is checked.
If any difference is detected, a violation is reported.
\end{description}
Whereas API checking is implemented inside CPM,
equivalence is more complex so that CurryCheck is used for this purpose.
To do so, CPM generates a set of new Curry programs
that are used by CurryCheck to perform the equivalence checks.
This is necessary since we want to compare
the behavior of some operation $f$ which is defined in two different
versions $v_1$ and $v_2$ of a package.
Thus, the modules of both package versions including all packages
on which these packages depend are copied
and renamed with the version number as a prefix.
For instance, a module \code{Mod} occurring in package version \code{2.1.4}
is copied and renamed into module \code{V\us{}2\us{}1\us{}4\us{}Mod}.
Thus, if there is an operation \code{f} occurring in module \code{Mod}
in package versions \code{2.1.4} and \code{2.2.1} to compare,
one can access both versions of this operation by the qualified
name \code{V\us{}2\us{}1\us{}4\us{}Mod.f}
and \code{V\us{}2\us{}2\us{}1\us{}Mod.f}.
After copying and renaming all modules,
CPM generates a new ``comparison'' module that
contains the following code:
\begin{curry}
import qualified V_2_1_4_Mod as V0
import qualified V_2_2_1_Mod as V1
  
test_Mod_f_Equivalent = V0.f <=> V1.f
\end{curry}
Now, this program can be passed to CurryCheck which tests
the equivalence as described in Sect.~\ref{sec:currycheck}.

This approach works only if both versions of operation \code{f}
have identical argument and result types.
However, this assumption is violated if \code{f} uses
data types defined in the package.
For instance, consider again Example~\ref{contrived-ground-equiv}
and assume that the operations are defined
with the same name but in two different versions of a package.
So we assume that a package contains in version \code{1.0.0}
a module \code{M} with the definitions:
\begin{curry}
data AB = A | B
data C  = C AB$\listline$
f x = C (h x)$\listline$
h A = A
\end{curry}
and in version \code{1.1.0} the same module but with these definitions:
\begin{curry}
data AB = A | B
data C  = C AB$\listline$
f A = C A
\end{curry}
After copying and renaming both versions as described above,
we might generate the following comparison module:
\begin{curry}
import qualified V_1_0_0_M as V0
import qualified V_1_1_0_M as V1$\listline$  
test_Mod_f_Equivalent = V0.f <=> V1.f
\end{curry}
This causes a type error since the types of
\code{V0.f} and \code{V1.f} differ:
\begin{curry}
V0.f :: V0.AB -> V0.C
V1.f :: V1.AB -> V1.C
\end{curry}
Fortunately, this is easy to fix.
Since the types are structurally identical (otherwise, the semantic versioning
is violated in the API comparison),
there exist bijective mappings between both renamed types.
Thus, CPM generates the code for these mappings:
\begin{curry}
t_AB :: V1.AB -> V0.AB
t_AB V1.A = V0.A
t_AB V1.B = V0.B$\listline$
t_C :: V1.C -> V0.C
t_C (V1.C x) = V0.C (t_AB x)
\end{curry}
By exploiting these mappings, CPM generates new operations
with modified types and a type-correct equivalence property:
\begin{curry}
M_f_1 :: V1.AB -> V0.C
M_f_1 x = t_C (V1.f x)$\listline$
M_f_2 :: V1.AB -> V0.C
M_f_2 x = V0.f (t_AB x)$\listline$
test_Mod_f_Equivalent = M_f_1 <=> M_f_2
\end{curry}
When CurryCheck tests this property, it finds the counter-example
to this equivalence as described in Sect.~\ref{sec:currycheck}.

Thanks to the techniques developed in this paper,
CPM's semantic versioning checking is a fully automatic process
and does not require any user interaction.

\subsection{Checking Implementations against their Specifications}

As discussed in \cite{AntoyHanus12PADL},
the distinctive features of Curry (e.g., non-deterministic operations,
demand-driven evaluation, functional patterns, set functions)
support writing executable high-level specifications
for a given problem.
By using purely functional features, e.g., sophisticated
data structures \cite{Okasaki98}, one can also write
efficient implementations for the same problem in Curry as well.
Thus, Curry can be used
as a wide-spectrum language for software development.
If a specification or contract is provided for some function,
one can exploit this information to support
run-time assertion checking with these specifications and contracts.
We can also use the same structures to check an implementation
against a given specification.
This can be done by exploiting equivalence checking as described
in the following.

We recall some notations introduced in \cite{AntoyHanus12PADL}.
A \emph{specification} for an operation $f$
is an operation \code{$f$'spec} of the same type as $f$.
A specification is typically written in a high-level manner
and less efficient than the actual implementation.
Nevertheless, it should be executable so that it can be used
to test the implementation in an automatic manner.
By contrast, a \emph{contract} can be weaker than a specification
and consists of a pre- and a postcondition.
If any of them is omitted, they are considered as always satisfied.
When they are explicitly defined, a \emph{precondition} for an operation $f$
of type $\tau \to \tau'$ is an operation:
\begin{curry}
$f$'pre :: $\tau$ ->$~$Bool
\end{curry}
restricting allowed argument values, whereas
a \emph{postcondition}\index{postcondition} for $f$
is an operation:
\begin{curry}
$f$'post :: $\tau$ ->$~\tau'$ ->$~$Bool
\end{curry}
which relates input and output values
(the generalization to operations with more than one argument
is straightforward).
A specification should precisely describe the meaning of an operation,
i.e., the declarative meaning of the specification and the implementation
of an operation should be equivalent.
Since they should be equivalent in any possible context,
equivalence in the sense of Def.~\ref{def:equivalence} is required.
By contrast, a contract is a partial specification, e.g.,
all results computed by the implementation should satisfy the
postcondition.

To discuss a concrete example, we consider the problem of sorting a list.
A high-level specification defines the result of sorting a given list
as a permutation of the input which is sorted.
Following Example~\ref{ex:permsort},
we define the following specification for the operation \code{sort}:
\begin{curry}
sort'spec :: [Int] -> [Int]
sort'spec xs | sorted ys = ys    where ys = perm xs
\end{curry}
To provide a simple implementation, we implement the quicksort algorithm
as follows:
\begin{curry}
sort :: [Int] -> [Int]
sort []     = []
sort (x:xs) = sort (filter (<x) xs) ++ [x] ++ sort (filter (>x) xs)
\end{curry}
Specifications and contracts are optional.
Nevertheless, they are useful in software development.
For instance, the Curry preprocessor transforms a program
with specifications and contracts so that they are used as
run-time assertions, as described in \cite{AntoyHanus12PADL}.
On the other hand, they can also be used statically
by testing them with various test inputs.
The equivalence of \code{sort} and \code{sort'spec} can be defined
as the following CurryCheck property:
\begin{curry}
sortSatisfiesSpecification = sort <=> sort'spec
\end{curry}
When this property is tested with the methods described in
Sect.~\ref{sec:currycheck},
CurryCheck reports that the above implementation of
\code{sort} is not equivalent to the specification \code{sort'spec}
for the example input \code{[0,0]}
(as the careful reader might have already noticed).

To automate this process, we have extended CurryCheck
so that it automatically generates such kinds of
equivalence properties when specifications are present
in a module to be tested.
For instance, consider a Curry module containing
a recursive specification of the factorial function
as well as an iterative implementation:
\begin{curry}
fac'spec :: Int -> Int
fac'spec n = if n==0 then 1
                     else n * fac'spec (n-1)$\listline$
fac :: Int -> Int
fac n = faci 1 1
 where faci m p = if m>n then p
                         else faci (m+1) (m*p)
\end{curry}
If we process this module with CurryCheck,
the following property is generated and tested:
\begin{curry}
facSatisfiesSpecification = fac <=> fac'spec
\end{curry}
However, this test does not terminate (or is terminated with a time out)
since \code{fac'spec} is not intended to be called with negative
integers. This intention can be expressed by the following precondition:
\begin{curry}
fac'spec'pre :: Int -> Bool
fac'spec'pre n = n >= 0
\end{curry}
Such a precondition is considered by CurryCheck when generating
test inputs, i.e., a generated input value is only accepted
for testing when it satisfies the given precondition.
If the implementation also contains a precondition,
both preconditions are taken into account.
Hence, after adding this precondition, CurryCheck
will successfully tests the equivalence of
the specification and implementation of \code{fac}.

As a final example of this section,
consider the non-deterministic list insertion.
The ``natural'' definition shown in Example~\ref{ex:insert}
is used as a specification and the implementation
is defined by disambiguating patterns in the left-hand sides
and a choice in the right-hand side of the non-trivial rule:
\begin{curry}
ndinsert'spec x ys     = x : ys
ndinsert'spec x (y:ys) = y : ndinsert'spec x ys$\listline$
ndinsert x []     = [x]
ndinsert x (y:ys) = x : y : ys  ?  y : ndinsert x ys
\end{curry}
Although the implementation seems to satisfy the specification,
CurryCheck reports an error if both input arguments are
\code{failed} and the result is computed up to head constructor form
(i.e., a constructor-rooted expression)
whereas:
\begin{curry}
ndinsert'spec failed failed
\end{curry}
has \code{failed:failed}
as a head normal form, \code{ndinsert failed failed} has no
head constructor form.
This shows that \code{ndinsert} does not satisfy its specification.

Although this example looks artificial, the introduction
of contextual equivalence checking for specifications
revealed an inconsistency in the standard module \code{Sort}
that went undetected for a long time.
This module contained various sort algorithms
(insertion sort, quick sort, merge sort, etc)
together with the following non-deterministic specification:
\begin{curry}
sort'spec xs | ys == perm xs && sorted ys = ys    where ys free
\end{curry}
Before the work described in this paper,
only ground equivalence of the implemented sort algorithms and the
specification were (successfully) tested.
After extending CurryCheck with contextual equivalence checking,
a difference was reported since this specification
is too strict: due to the strict equality in the condition
(\code{ys$\;$==$\;$perm$\;$xs}), this specification
does not yield any result for the expression:
\begin{curry}
null (permSort [failed])
\end{curry}
whereas the implementations return the result \code{False}.
CurryCheck revealed this non-equivalence by testing
the specification and each implementation with the
partial one-element input list \code{[$\bot$]}.
This behavioral difference was then fixed by relaxing the specification
as shown in Example~\ref{ex:permsort}.

\section{Towards Verification of Equivalences}
\label{sec:verification}

In the previous section we have developed a framework
to use a property-based test tool to check the
equivalence of operations.
This is useful since the equivalence checking process is fully automatic.
Although testing cannot show the absence of errors,
it is an important step in software development
since it can find errors and, if no errors are detected,
provides confidence in the developed software.
The next step to increase confidence is verification.
Although it is manual process, we want to demonstrate
in this section that our results are also helpful
to support verification of equivalence properties.

As a concrete example, we want to show the correctness
of an implementation of a sort algorithm
known as \emph{straight selection sort}~\cite{knuth1998art}.
Informally, a list is sorted by selecting its smallest element,
sorting the remaining elements, and placing the smallest element
in front of the sorted remaining elements.
Thus, the algorithm performs two operations:
(1) the selection of a smallest element and
(2) the computation of the remaining list without this element.
A very simple implementation of this algorithm performs
two traversals of the input list, one for each operation.
A more efficient algorithm performs a single traversal
by means of an operation, \code{minRest}, which selects the smallest
element and the remaining ones simultaneously but requires an
additional accumulator argument.
This results in the following implementation,
where the accumulator is the second argument of \code{minRest}:
\begin{curry}
sort []     = []
sort (x:xs) = m : sort r
 where (m,r) = minRest x [] xs$\listline$
minRest m rs []     = (m,rs)
minRest m rs (y:ys) = if m<=y then minRest m (y:rs) ys
                              else minRest y (m:rs) ys
\end{curry}
This implementation is more efficient than a direct implementation
with two list traversals in each call to \code{sort},
but it is also more complicated so that its correctness is not as apparent.
Therefore, we can apply CurryCheck to test the equivalence
of the implementation and the specification:
\begin{curry}
sort'spec xs | sorted ys = ys    where ys = perm xs
\end{curry}
Surprisingly, CurryCheck reports a counter-example which is
due to lazy pattern matching of the \code{where} clause.
\code{sort} returns a constructor-rooted term for non-empty lists
if the arguments are not evaluated, e.g., the expression:
\begin{curry}
null (sort (failed:failed))
\end{curry}
evaluates to \code{False}, whereas the same expression
with \code{sort} replaced by \code{sort'spec} has no value.

In order to obtain an implementation equivalent to the
specification, we enforce the matching on the result of \code{minRest}
by changing the definition of \code{sort} to:
\begin{curry}
sort []     = []
sort (x:xs) = case minRest x [] xs of (m,r) -> m : sort r
\end{curry}
Now CurryCheck does not find any counter-example showing
the non-equivalence of \code{sort} and \code{sort'spec}
so that we can try to verify this property.
For this purpose, we prove the following lemma about the auxiliary
operation \code{minRest}.

\newcommand{\conc}{\code{++}} 
\newcommand{\perm}{\code{perm}\xspace} 
\newcommand{\sort}{\code{sort}\xspace} 
\newcommand{\sorted}{\code{sorted}\xspace} 
\newcommand{\sortspec}{\code{sort'spec}\xspace} 
\newcommand{\minrest}{\code{minRest}\xspace} 
\newcommand{\post}[1]{#1^{post}} 

\begin{lemma}
\label{lemma:mr-sound}
Let $x$ be an integer and $rs$ and $xs$ lists of integers
such that $x \leq z$ for all $z \in rs$.
Then, for some list of integers $r$,
$\minrest~x~rs~xs \tostar (m,r)$
with $m \leq z$ for all $z \in (x:rs\conc xs)$
and $\perm~(x:rs \conc xs) \tostar m:r$ (where $rs \conc xs$
denotes the concatenation of the lists $rs$ and $xs$).
\end{lemma}
\begin{proof}
We assume that $x \leq z$ for all $z \in rs$.
We prove the lemma by induction on the length of the list $xs$.

Base case: $xs = []$:
Since $\minrest~x~rs~[] \tostar (x,rs)$,
$x \leq z$ for all $z \in (x:rs\conc [])$ follows from the assumption
and $\perm~(x:rs\conc []) \tostar x:rs$ obviously holds
by definition of \perm.

Inductive case: $xs = y:ys$:
We assume that the claim holds for the list $ys$.
We consider the two cases occurring in the definition of $\minrest$.

If $x \leq y$, then
$\minrest~x~rs~(y:ys) \tostar \minrest~x~(y:rs)~ys$.
Since $x \leq z$ for all $z \in (y:rs)$,
we can apply the induction hypothesis to the last expression and obtain
$\minrest~x~(y:rs)~ys \tostar (m,r)$
with $x \leq z$ for all $z \in (x:y:rs\conc ys)$
and $\perm~(x:y:rs\conc ys) \tostar m:r$.
Hence $m:r$ is also a permutation of $(x:rs\conc y:ys)$ and
$m \leq z$ for all $z \in (x:rs\conc y:ys)$.

If $x > y$, then
$\minrest~x~rs~(y:ys) \tostar \minrest~y~(x:rs)~ys$.
Since $y \leq z$ for all $z \in (x:rs)$,
we can apply the induction hypothesis to the last expression and obtain
$\minrest~y~(x:rs)~ys \tostar (m,r)$
with $m \leq z$ for all $z \in (y:x:rs\conc ys)$
and $\perm~(y:x:rs\conc ys) \tostar m:r$.
Hence $m:r$ is also a permutation of $(x:rs\conc y:ys)$ and
$m \leq z$ for all $z \in (x:rs\conc y:ys)$.
\end{proof}
Now we can prove the soundness of \sort, i.e., the fact
that each value computed by \sort can also be derived by \sortspec.
\begin{lemma}
\label{lemma:sort-sound}
Let $l$ be a list of $n$ integers, for $n \geq 0$,
and $\sort~l \tostar l'$ for some value $l'$.
Then $\sortspec~l \tostar l'$.
\end{lemma}
\begin{proof}
We prove the lemma by induction on the length of the list $l$.

Base case: $l = []$:
Since $\sort~[] \to []$ and $\sortspec~[] \tostar []$,
the claim obviously holds.

Inductive case: $l = x:xs$:
By definition of $\sort$, $\sort~l \tostar m:ys$ with
$\minrest~x~[]~xs \tostar (m,r)$ and $\sort~r \tostar ys$.
Lemma~\ref{lemma:mr-sound} implies
(1) $perm~(x:xs) \tostar m:r$ and
(2) $m \leq z$ for all $z \in (x:xs)$.
Since $(m:r)$ is a permutation of $(x:xs)$,
$m$ is an element of $l$ and $r$ is a list shorter than $l$.
Thus, we can apply the induction hypothesis to $r$
and we obtain $\sortspec~r \tostar ys$.
Hence, by definition of \sortspec,
(3) $perm~r \tostar ys$ and (4) $sorted~ys \tostar \true$.
Properties (1) and (3) imply that $(m:ys)$ is a permutation of $(x:xs)$.
Together with (2), $m$ is smaller than or equal to all elements of $(m:ys)$
so that, due to (4), $sorted~(m:ys) \tostar \true$.
Therefore, $\sortspec~(x:xs) \tostar m:ys$.
\end{proof}
Next we prove the completeness of \sort, i.e., the fact
that each value derived by \sortspec can also be computed by \sort.
\begin{lemma}
\label{lemma:sort-complete}
Let $l$ be a list of $n$ integers, for $n \geq 0$,
and $\sortspec~l \tostar l'$ for some value $l'$.
Then $\sort~l \tostar l'$.
\end{lemma}
\begin{proof}
Assume that $\sortspec~l \tostar l'$ for some value $l'$.
Since it is easy to check that both \sort and $\minrest$ are totally
defined and terminating, there exists a value $l''$ with
$\sort~l \tostar l''$.
By Lemma~\ref{lemma:sort-sound},
$\sortspec~l \tostar l''$.
By definition of \sortspec, both $l'$ and $l''$
are permutations of $l$ that are sorted.
To prove the claim, we show $l'=l''$ by induction on the length $n$ of $l$.

Base case: $n=0$: Then $l'=[]=l''$ by definition of \sortspec.

Inductive case: $n>0$:
Since both $l'$ and $l''$ are permutations of $l$,
they must be non-empty so that they have the form
$l' = x:xs$ and $l'' = y:ys$.
Since $l'$ and $l''$ are sorted, $x$ and $y$ are smallest elements of $l$.
Thus $x=y$ and $xs$ and $ys$ are sorted and permutation of $l$
without element $x$. Hence, we can apply the induction hypothesis.
\end{proof}
~
Now we can finally prove that our sort implementation
is equivalent to its specification.
Similarly to property testing, we exploit the refined equivalence criteria,
in particular, Cor.~\ref{cor:equal-partial-value-equivalence},
in the proof.
\begin{corollary}
\label{cor:sort-correct}
\sort and \sortspec are equivalent.
\end{corollary}
\begin{proof}
First, we prove that \sort and \sortspec computes the same set of
partial values if the argument is a list of integers $l$.

Let $t$ be a partial value with $sort~l \tostar t$.
Since both \sort and $\minrest$ are totally defined and terminating
and these are the only operations occurring in this derivation,
we can replace in the derivation $sort~l \tostar t$
each application of rule \rulename{Bot} by a sequence of
\rulename{Fun} rules
deriving the same subexpression to some value so that
we obtain a derivation $\sort~l \tostar l'$ where $l'$ is a value.
By construction, $t$ must be smaller than $l'$ w.r.t.\ information ordering
and $l' \tostar t$ (by definition of $\tostar$).
By Lemma~\ref{lemma:sort-sound}, $\sortspec~l \tostar l'$.
Hence, $\sortspec~l \tostar l' \tostar t$.

To prove the other direction, let $t$ be a partial value
with $\sortspec~l \tostar t$.
If $t$ is a value, Lemma~\ref{lemma:sort-complete} implies $\sort~l \tostar t$.
Hence, assume that $t$ has at least one occurrence of $\bot$.
We can distinguish the following cases for $t$:
\begin{description}
\item[$t = \bot$]
This case is trivial since $sort~l \to \bot$ always holds by definition
of $\tostar$.
\item[$t = {[\bot]}$]
By definition of $\sortspec$, the result list
is a permutation of the input list so that $l$ must be a one-element list,
i.e., $l = [i]$ for some integer $i$.
Thus, $\sort~l \tostar l \tostar [\bot]$ by definition of
\sort and $\tostar$.
\item[$t \neq \bot \land t \neq {[\bot]}$]
By definition, \sortspec yields a constructor-rooted result only
if $\sorted~l' \tostar \code{True}$ for some partial value $l'$
with $\perm~l \tostar l'$.
The case that $l'$ contains less than two elements can be treated as before.
If $l'$ contains at least two elements,
$\sorted~l' \tostar \code{True}$ implies that $l'$ must be a value,
i.e., a list of integers (otherwise $\sorted~l' \tostar \bot$
by definition of \sorted).
Thus, $\sortspec~l \tostar l' \tostar t$.
By Lemma~\ref{lemma:sort-complete},
$\sort~l \tostar l'$ so that $\sort~l \tostar t$.
\end{description}
Now consider an argument $xs$ which is partial value but not a value,
i.e., $xs$ contains occurrences of $\bot$.
We can distinguish the following cases:
\begin{description}
\item[$xs = \bot$]
Since both \sortspec and \sort are defined on a case distinction
of the argument list $xs$, $\bot$ is the only result.
\item[$xs = x : \bot$] Then both \sortspec and \sort yield $\bot$
as the only computable result (note that pattern matching
enforced by \code{case} in the rule of \code{sort} is
important here!).
\item[{$xs = \bot : []$}] Then both \sortspec and \sort yield
$[\bot]$ or some smaller partial value (w.r.t.\ information ordering).
\item[$xs = x : y : ys$] where $x$, $y$, or $ys$ contain $\bot$:
Since the list $xs$ contains at least two elements,
both \sortspec and \sort consider the complete list
in order to produce some head normal form.
Thus, $\sortspec~xs \tostar \bot$ and $\sort~xs \tostar \bot$ are the
only computable results.
\end{description}
Altogether, we have shown that
$\sortspec~xs \tostar t'$ iff $\sort~xs \tostar t'$
for all partial values $xs$ and $t'$.
Therefore, \sortspec and \sort are equivalent
by Cor.~\ref{cor:equal-partial-value-equivalence}.
\end{proof}
The proof in this section demonstrates a general strategy
to verify the equivalence of operations:
\begin{enumerate}
\item Show ground equivalence.
\item Consider the extension to partial values and apply
      Cor.~\ref{cor:equal-partial-value-equivalence}.
\end{enumerate}
Of course, this two-step strategy might not work for all operations.
For instance, non-terminating operations, like \code{ints1} and \code{ints2},
are trivially ground equivalent.

\section{Related Work}
\label{sec:related}

Equivalence of operations was defined for functional logic programs
in \cite{AntoyHanus12PADL}. There, this notion is applied
to relate specifications and implementations.
Moreover, it is shown how to use specifications as dynamic contracts
to check the correct behavior of implementations at run-time,
but static methods to check equivalence are not discussed.

Bacci et al.\ \cite{BacciEtAl12} formalized various notions
of equivalence, reviewed in Sect.~\ref{sec:intro},
and developed the tool AbsSpec that, from a Curry program,
derives specifications,
i.e., equations up to some fixed depth of the involved
expressions.
Although the derived specifications are equivalent
to the implementation, their method cannot be used
to check the equivalence of arbitrary operations.

QuickSpec \cite{ClaessenSmallboneHughes10} has similar goals
as AbsSpec but is based on a different setting.
QuickSpec infers specifications in the form of equations
from a given functional program but it uses a black box approach,
i.e., it uses testing to infer program properties.
Thus, it can be seen as an intermediate approach
between AbsSpec and our approach:
similarly to our approach, QuickSpec uses property-based testing
to check the correctness of specifications,
but it is restricted to functional programs,
which simplifies the notion of equivalence.

Our method to check equality of computed results for all partial
values is also related to testing properties in non-strict
functional languages \cite{DanielssonJansson04}.
Thanks to the non-deterministic features of Curry,
our approach does not require impure features
like \code{isBottom} or \code{unsafePerformIO},
which are used in \cite{DanielssonJansson04} to compare partial values.

Partial values as inputs for property-based testing
are also used in Lazy SmallCheck \cite{RuncimanNaylorLindblad08},
a test tool for Haskell which generates data in a systematic
(rather than random) manner. Partial input values are used
to reduce the number of test cases: if a property is satisfied
for a partial value, it is also satisfied for all refinements
of this partial value so that it is not necessary to test
these refinements. Thus, Lazy SmallCheck exploits partial values
to reduce the number of test cases for total values,
where in our approach partial values are used to avoid
testing with all possible contexts and to find counter examples
which might not be detected with total values only.
In contrast to our explicit encoding of partial values,
which is possible due to the logic features of Curry,
Lazy SmallCheck represents partial values as run-time errors
which are observed using imprecise exceptions
\cite{PeytonJonesEtAl99}.

As discussed in Sect.~\ref{sec:strictness},
equivalent strictness properties are a necessary
condition for contextual equivalence
and strictness information can be used to improve
equivalence checking.
Interestingly, many recent tools to analyze
the strictness behavior of functions use property-based testing
approaches for this purpose, where impure features,
as in \cite{DanielssonJansson04},
are used to deal with partial values.
For instance, the tool StrictCheck \cite{Chitil11TR}
tries to find non-least-strict functions, i.e.,
functions which can be made less strict.
However, the tool might yield false positives as well as false negatives
so that the programmer has to check the reported functions in detail.
Polymorphic functions with a minimal strictness behavior
can be analyzed with the Sloth tool \cite{ChristiansenSeidel11}
which generates examples if this behavior is not ensured.
A more recent tool, also called StrictCheck \cite{FonerZhangLampropoulos18},
supports a language where the programmer can specify
intended strictness demands of functions.
These demands are tested by catching run-time errors caused
by using undefined values similarly to Lazy SmallCheck.
The objectives of these approaches are different
from our work, since changing the strictness behavior
of functions leads to non-equivalent operations.
Nevertheless, the difference in strictness properties,
shown in the examples in these papers,
can also be detected by our methods
(see the collection of examples distributed with CurryCheck).

The use of property-based testing
to check the equivalence of operations in a software package manager
with support for semantic versioning is proposed in \cite{Hanus17ICLP}.
This approach focuses on ensuring the termination
of equivalence checking by introducing the notion
of productive operations. However, for terminating operations
only ground equivalence is tested so that the proposed semantic versioning
checking method is more restrictive than ours.
We have shown in Sect.~\ref{sec:cpm} how the
results presented in this paper can be used
to improve this initial semantic versioning tool.

\section{Conclusions}

We have presented a method to check the equivalence of
operations defined by a functional logic program.
This method is useful for software package managers
to provide automatic semantic versioning checks,
i.e., to compare two different versions of a software package,
or to check the correctness of an implementation against a specification.
Since we developed our results for a non-strict functional logic language,
the same techniques can be used to test equivalence
in purely functional languages, e.g., for Haskell programs.

We have shown that the general equivalence of operations,
which requires that the same values are computed in all possible contexts,
can be reduced to checking or proving equality of partial
results expressions.
Our results support the use of automatic property-based test tools
for equivalence checking.
Although this method is incomplete, i.e.,
it does not prove equivalence,
it may disprove it.
This knowledge saves investing time in futile attempts
to formally prove the equivalence of non-equivalent functions.
It also warns the programmer of a defect,
if her intent was to code a function equivalent to another.
Moreover, the presented results could also be helpful for
manual proof construction or using proof assistants.

For future work, it is interesting to explore how
automatic theorem provers can be used to verify
specific equivalence properties.

\paragraph{Acknowledgments.}
The authors are grateful to Finn Teegen for constructive remarks
to an initial version of this paper, and to
the anonymous reviewers for their helpful comments to improve this paper.
This material is based in part upon work supported 
by the National Science Foundation under Grant No.~1317249.


\end{document}